\documentclass{amsart}
\usepackage{amssymb,url,tikz,enumerate,rotating,color}
\newcommand{\todo}[1]{{#1}}

\numberwithin{equation}{section}
\theoremstyle{plain}
\newtheorem{theorem}[equation]{Theorem}
\newtheorem*{theorem*}{Theorem}
\newtheorem{lemma}[equation]{Lemma}
\newtheorem{proposition}[equation]{Proposition}
\newtheorem{corollary}[equation]{Corollary}
\newtheorem*{corollary*}{Corollary}
\theoremstyle{definition}
\newtheorem{definition}[equation]{Definition}
\newtheorem{example}[equation]{Example}

\usetikzlibrary{matrix,cd}
\tikzstyle{dot}=[circle, draw=black, fill=black!25, inner sep=.4ex]
\newif\ifvflip\pgfkeys{/tikz/vflip/.is if=vflip}
\newif\ifhflip\pgfkeys{/tikz/hflip/.is if=hflip}
\newif\ifhvflip\pgfkeys{/tikz/hvflip/.is if=hvflip}
\newlength\morphismheight
\newlength\wedgewidth
\tikzset{width/.initial=1mm}
\makeatletter
\tikzstyle{morphism}=[font=\small,morphismshape]
\pgfdeclareshape{morphismshape}
{
    \savedanchor\centerpoint
    {
        \pgf@x=0pt
        \pgf@y=0pt
    }
    \anchor{center}{\centerpoint}
    \anchorborder{\centerpoint}
    \saveddimen\overallwidth
    {
        \pgfkeysgetvalue{/tikz/width}{\minwidth}
        \pgf@x=\wd\pgfnodeparttextbox
        \ifdim\pgf@x<\minwidth
            \pgf@x=\minwidth
        \fi
    }
    \savedanchor{\upperrightcorner}
    {
        \pgf@y=.5\ht\pgfnodeparttextbox
        \advance\pgf@y by -.5\dp\pgfnodeparttextbox
        \pgf@x=.5\wd\pgfnodeparttextbox
    }
    \anchor{north}
    {
        \pgf@x=0pt
        \pgf@y=0.5\morphismheight
    }
    \anchor{north east}
    {
        \pgf@x=\overallwidth
        \multiply \pgf@x by 2
        \divide \pgf@x by 5
        \pgf@y=0.5\morphismheight
    }
    \anchor{east}
    {
        \pgf@x=\overallwidth
        \divide \pgf@x by 2
        \advance \pgf@x by 5pt
        \pgf@y=0pt
    }
    \anchor{west}
    {
        \pgf@x=-\overallwidth
        \divide \pgf@x by 2
        \advance \pgf@x by -5pt
        \pgf@y=0pt
    }
    \anchor{north west}
    {
        \pgf@x=-\overallwidth
        \multiply \pgf@x by 2
        \divide \pgf@x by 5
        \pgf@y=0.5\morphismheight
    }
    \anchor{south east}
    {
        \pgf@x=\overallwidth
        \multiply \pgf@x by 2
        \divide \pgf@x by 5
        \pgf@y=-0.5\morphismheight
    }
    \anchor{south west}
    {
        \pgf@x=-\overallwidth
        \multiply \pgf@x by 2
        \divide \pgf@x by 5
        \pgf@y=-0.5\morphismheight
    }
    \anchor{south}
    {
        \pgf@x=0pt
        \pgf@y=-0.5\morphismheight
    }
    \anchor{text}
    {
        \upperrightcorner
        \pgf@x=-\pgf@x
        \pgf@y=-\pgf@y
    }
    \backgroundpath
    {
    \begin{scope}
        \pgfkeysgetvalue{/tikz/fill}{\morphismfill}
        \pgfsetstrokecolor{black}
        \pgfsetlinewidth{.7pt}
        \begin{scope}
        \pgfsetstrokecolor{black}
        \pgfsetfillcolor{white}
        \pgfsetlinewidth{.7pt}
                \ifhflip
                    \pgftransformyscale{-1}
                \fi
                \ifvflip
                    \pgftransformxscale{-1}
                \fi
                \ifhvflip
                    \pgftransformxscale{-1}
                    \pgftransformyscale{-1}
                \fi
                \pgfpathmoveto{\pgfpoint
                    {-0.5*\overallwidth-5pt}
                    {0.5*\morphismheight}}
                \pgfpathlineto{\pgfpoint
                    {0.5*\overallwidth+5pt}
                    {0.5*\morphismheight}}
                \pgfpathlineto{\pgfpoint
                    {0.5*\overallwidth + \wedgewidth}
                    {-0.5*\morphismheight}}
                \pgfpathlineto{\pgfpoint
                    {-0.5*\overallwidth-5pt}
                    {-0.5*\morphismheight}}
                \pgfpathclose
                \pgfusepath{fill,stroke}
        \end{scope}
    \end{scope}
    }
}
\makeatother
\newcommand{\tinymult}[1][dot]{
\smash{\raisebox{-2pt}{\hspace{-5pt}\ensuremath{\begin{pic}[scale=0.4,yscale=-1]
    \node (0) at (0,0) {};
    \node[#1, inner sep=1.5pt] (1) at (0,0.55) {};
    \node (2) at (-0.5,1) {};
    \node (3) at (0.5,1) {};
    \draw (0.center) to (1.center);
    \draw (1.center) to [out=left, in=down, out looseness=1.5] (2.center);
    \draw (1.center) to [out=right, in=down, out looseness=1.5] (3.center);
    \node[#1, inner sep=1.5pt] (1) at (0,0.55) {};
\end{pic}
}\hspace{-3pt}}}}
\newcommand{\tinyunit}[1][dot]{
\smash{\raisebox{1pt}{\hspace{-3pt}\ensuremath{\begin{pic}[scale=0.4,yscale=-1]
    \node (0) at (0,0) {};
    \node[#1, inner sep=1.5pt] (1) at (0,0.55) {};
    \draw (0.center) to (1.north);
\end{pic}
}\hspace{-1pt}}}}
\newcommand\sxto[1]{\mathbin{\smash{
\begin{tikzpicture}[baseline={([yshift=-1pt]
current bounding box.south)}]
    \node (A) at (0,0) [inner xsep=0pt, inner ysep=1pt, minimum width=0.15cm] {\ensuremath{\scriptstyle #1}};
    \draw [->, line width=0.4pt, line cap=round]
        ([xshift=-2.5pt] A.south west)
        to ([xshift=3pt] A.south east);
\end{tikzpicture}}}}

\newenvironment{pic}[1][]
{\begin{aligned}\begin{tikzpicture}[font=\tiny,#1]}
{\end{tikzpicture}\end{aligned}}

\newcommand{\cat}[1]{\ensuremath{\mathbf{#1}}}
\newcommand{\op}{\ensuremath{{}^{\text{op}}}}
\newcommand{\id}[1][]{\ensuremath{\mathrm{id}_{#1}}}
\newcommand{\inprod}[2]{\ensuremath{\langle #1 \,|\, #2 \rangle}}
\DeclareMathOperator{\st}{st}
\DeclareMathOperator{\cst}{cst}
\DeclareMathOperator{\dst}{dst}
\DeclareMathOperator{\Tr}{Tr}

\begin{document}
  \title{Reversible Monadic Computing} 
  \author{Chris Heunen}
  \address{Department of Computer Science, University of Oxford, United Kingdom}
  \email{heunen@cs.ox.ac.uk}
  \author{Martti Karvonen}
  \address{Department of Mathematics and Systems Analysis, Aalto University, Finland}
  \email{martti.karvonen@aalto.fi}
  \thanks{Supported by the Engineering and Physical Sciences Research Council Fellowship EP/L002388/1. We thank an anonymous referee for Example~\ref{ex:emnonfem}, and Jorik Mandemaker, Sean Tull, and Maciej Pirog for helpful discussions.}
 \begin{abstract}
  We extend categorical semantics of monadic programming to reversible computing, by considering monoidal closed dagger categories: the dagger gives reversibility, whereas closure gives higher-order expressivity. 
  We demonstrate that Frobenius monads model the appropriate notion of coherence between the dagger and closure by reinforcing Cayley's theorem; by proving that effectful computations (Kleisli morphisms) are reversible precisely when the monad is Frobenius; by characterizing the largest reversible subcategory of Eilenberg--Moore algebras; and by identifying the latter algebras as measurements in our leading example of quantum computing. Strong Frobenius monads are characterized internally by Frobenius monoids. 
\end{abstract}
\maketitle

\section{Introduction} 

The categorical concept of a \emph{monad} has been tremendously useful in programming, as it extends purely functional programs with nonfunctional effects. For example, using monads one can extend a functional programming language with  nondeterminism, probabilism, stateful computing, error handling, read-only environments, and input and output~\cite{wadler:monads}. Haskell incorporates monads in its core language. On the theoretical side, there are satisfyingly clean categorical semantics. Simply typed $\lambda$-calculus, that may be regarded as an idealized functional programming language, takes semantics in Cartesian closed categories~\cite{lambekscott:categoricallogic}. The functional programming concept of a monad is modeled by the categorical concept of a monad~\cite{moggi:monads}.

In classical computation it is not always possible to reconstruct the input to an algorithm from its output. However, by using auxiliary bits, any classical computation can be turned into a \emph{reversible} one~\cite{toffoli:reversible}. Such a computation uses invertible primitive gates, and composition preserves invertibility. As discarding information requires work, reversible computations could in principle be implemented at higher speeds. The only operation costing power is the final discarding of auxiliary bits.

This is brought to a head in \emph{quantum computing}, where any deterministic evolution of quantum bits is invertible, unlike the eventual measurement that converts quantum information to classical information. Another novelty in quantum computing is that it is impossible to copy or delete quantum information. This leads to a linear type theory of resources rather than a classical one~\cite{szabo:proofs}: quantum computing takes semantics in monoidal categories, rather than Cartesian ones~\cite{abramskycoecke:categoricalsemantics}.

Led by quantum computing, this article extends the categorical semantics of monadic programming to reversible computing. To allow for a linear type theory we consider monoidal closed categories. To allow for reversible computations, we consider \emph{dagger categories}; in general these correspond to bidirectional computations rather than invertible ones, which in the quantum case comes down to the same thing. To allow for monadic effects, we introduce \emph{Frobenius monads}. In the presence of a dagger, any monad gives rise to a comonad; a Frobenius monad is one that interacts with its comonad counterpart via the following \emph{Frobenius law}:
\begin{equation}\label{eq:frobeniuslaw}
  \begin{pic}
   \draw (0,0) to (0,1) to[out=90,in=180] (.5,1.5) to (.5,2);
   \draw (.5,1.5) to[out=0,in=90] (1,1) to[out=-90,in=180] (1.5,.5) to (1.5,0);
   \draw (1.5,.5) to[out=0,in=-90] (2,1) to (2,2);
   \node[dot] at (.5,1.5) {};
   \node[dot] at (1.5,.5) {};
  \end{pic}
  \quad = \quad
  \begin{pic}[xscale=-1]
   \draw (0,0) to (0,1) to[out=90,in=180] (.5,1.5) to (.5,2);
   \draw (.5,1.5) to[out=0,in=90] (1,1) to[out=-90,in=180] (1.5,.5) to (1.5,0);
   \draw (1.5,.5) to[out=0,in=-90] (2,1) to (2,2);
   \node[dot] at (.5,1.5) {};
   \node[dot] at (1.5,.5) {};
  \end{pic}
\end{equation}
Here we used the graphical calculus for monoidal categories~\cite{selinger:graphicallanguages,marsden:diagrams}, that will be explained further in Section~\ref{sec:dagger}, along with several examples.\footnote{We often need to reason simultaneously about morphisms \emph{in} a monoidal category and endofunctors \emph{on} it. Unfortunately there is no sound and complete graphical proof calculus that would handle this yet. Therefore we cannot use the graphical calculus exclusively and also have to use traditional commutative diagrams.}

Our main contribution is to take reversal as a primitive and so justify the claim that Frobenius monads are precisely the right notion as follows:
\begin{itemize}

  \item Section~\ref{sec:monoids} justifies the Frobenius law as a necessary (and sufficient) consequence of coherence between the dagger and closure. In a reversible setting, it is natural to consider \emph{involutive} monoids. In a monoidal closed category, any monoid embeds into a canonical one by Cayley's theorem. We prove that this embedding preserves the involution induced by the dagger if and only if the monoid satisfies the Frobenius law. 
  This derivation from first principles is a noncommutative generalization \todo{of}~\cite[Theorem~4.3]{pavlovic:abstraction} with a new proof.

  \item Section~\ref{sec:monads} characterizes Frobenius monads \emph{internally}. Monads are an \emph{external} notion. A good example is the writer monad, that allows programs to keep auxiliary output alongside the computation. These values accumulate according to some monoid. Any monoid gives rise to a strong monad, and Frobenius monoids give rise to strong Frobenius monads. In general this is merely an adjunction and not an equivalence, but we work out that the converse holds in the Frobenius setting. This is a noncommutative generalization of~\cite[Corollary~4.5]{pavlovic:abstraction}. It also generalizes the classic Eilenberg--Watts theorem from homological algebra to categories that are not necessarily abelian. As Frobenius monoids satisfy the very same law~\eqref{eq:frobeniuslaw} as Frobenius monads, only interpreted in a category rather than by endofunctors on it, this also exhibits that reversible settings are closed under categorification. 

  \item We show that the extension of reversible pure computations with effects modeled by a monad results in reversible effectful computations if and only if the monad is a Frobenius monad. More precisely, Section~\ref{sec:kleisli} shows that a monad on a dagger category is a Frobenius monad if and only if the dagger extends to the category of Kleisli algebras. This reinforces that Frobenius monads model the right notion of effects for reversible computing.
  Section~\ref{sec:fem} identifies the largest subcategory of all algebras with this property, which we call \emph{Frobenius--Eilenberg--Moore} algebras. Section~\ref{sec:measurement} exemplifies them in the quantum setting by arguing that they correspond precisely to measurements via \emph{effect handlers}~\cite{plotkinpretnar:handlers}.
\end{itemize}

Frobenius monads have been studied before~\cite{street:ambidextrous,lauda:ambidextrous}, and monads have been used as semantics for quantum computing before~\cite{hasuohoshino:quantum,altenkirchgreen:monad,altenkirchgrattage:qml}, but not in a dagger setting, except for~\cite{pavlovic:abstraction} that deals with the commutative case abstractly. Conversely, reversible programming has been modeled in dagger categories~\cite{bowmanetal:traced}, but not using monads. Daggers and monads have come together in coalgebra before~\cite{jacobs:involutive,jacobs:walks}; the same holds for quantum programming languages~programming languages~\cite{greenetal:quipper,selingervaliron:lambda}, and matrix algebra~\cite{devosdebaerdemacker:matrix}. The current work differs by systematically starting from first principles. 
We intend to fit probabilistic programming into this setup in future work.

\section{Dagger categories}\label{sec:dagger}

Let us model types as objects $A,B,C,\ldots $in a category, and computations as morphisms $f,g,h,\ldots$. To model composite types, we consider \emph{monoidal categories}, where one can not only compose computations in sequence $A \sxto{f} B \sxto{g} C$, but also in parallel $A \otimes B \sxto{f \otimes g} C \otimes D$. This much is standard~\cite{barrwells:computer}. To model \emph{reversible} computations, we need an operation turning a computation $A \sxto{f} B$ into a computation $B \to A$, such that reversing twice doesn't do anything. 

\begin{definition}
  A \emph{dagger} is a functor $\dag\colon\cat{C}\op\to\cat{C}$ satisfying $A^\dag=A$ on objects and $f^{\dag\dag}=f$ on morphisms. A \emph{dagger category} is a category equipped with a dagger. 
\end{definition}

Dagger categories can behave quite different from ordinary (non-dagger) ones, see \textit{e.g.}~\cite[9.7]{hott}. They are especially useful as semantics for quantum computing~\cite{heunenvicary:cqm}. Note that \emph{reversible} computing does not mean computations are \emph{invertible}. An invertible morphism $f$ in a dagger category is \emph{unitary} when $f^\dag=f^{-1}$. Similarly, an endomorphism $f$ is \emph{self-adjoint} when $f=f^\dag$.
As a rule, any structure in sight should cooperate with the dagger. 

\begin{definition}
  A monoidal category is called a \emph{monoidal dagger category} when $(f\otimes g)^\dag=f^\dag \otimes g^\dag$, and all coherence isomorphisms $A \otimes (B \otimes C) \sxto{\alpha} (A \otimes B) \otimes C$, $I \otimes A \sxto{\lambda} A$, and $A \otimes I \sxto{\rho} A$, are unitary. In a \emph{symmetric monoidal dagger category} additionally the swap maps $A \otimes B \sxto{\sigma} B \otimes A$ are unitary.
\end{definition}

We will mainly consider the following two examples.

\begin{example}
  The symmetric monoidal dagger category $\cat{Rel}$ has sets as objects. Morphisms $A \to B$ are \emph{relations} $R \subseteq A \times B$, with composition $S \circ R = \{ (a,c) \mid \exists b \colon (a,b) \in R, (b,c) \in S \}$. The dagger is given by $R^\dag = \{(b,a) \mid (a,b \in R)\}$, and the monoidal structure is given by Cartesian products. We may think of $\cat{Rel}$ as modeling \emph{nondeterministic computation}~\cite{jacobsrutten:coalgebra}.
\end{example}


\begin{example}
  The symmetric monoidal dagger category $\cat{FHilb}$ has finite-dimensional complex \emph{Hilbert spaces} as objects and linear maps as morphisms. The dagger is given by adjoints: $f^\dag$ is the unique linear function satisfying $\inprod{f(x)}{y} = \inprod{x}{f^\dag(y)}$; in terms of matrices it is the conjugate transpose. The monoidal structure is given by tensor products of Hilbert spaces. This models \emph{quantum computation}~\cite{abramskycoecke:categoricalsemantics}.
\end{example}

There are many other examples. Reversible probabilistic computation is modelled by the category of doubly stochastic maps~\cite[2.3.5]{coeckepaquettepavlovic:structuralism}; this generalizes to labelled Markov chains~\cite{panangaden:markov}. Universal constructions can generate examples with specific properties~\cite{pavlovic:toymodels}. Finally, one can formally add daggers to a category in a free or cofree way~\cite[3.1.17 and~3.1.19]{heunen:thesis}. We will be interested in the following way to turn a monoidal dagger category into a new one of endofunctors on the old one. It could be regarded as modeling \emph{second-order computation}, because the computations in the new category may refer to computations in the old one (but not to themselves).

\begin{example}
  A functor $\cat{C}\sxto{F}\cat{D}$ between dagger categories is a \emph{dagger functor} when $F(f^\dag)=F(f)^\dag$ on morphisms.
  Let $\cat{C}$ be a monoidal dagger category.
  If $F(A) \sxto{\beta_A} G(A)$ is a natural transformation between dagger functors $\cat{C} \sxto{F,G} \cat{C}$, then so is $G(A) \sxto{\beta_A^\dag} F(A)$.
  Thus the category $[\cat{C},\cat{C}]_\dag$ of dagger functors $\cat{C}\to\cat{C}$ is again a monoidal dagger category
  by $G \otimes F = G \circ F$.
\end{example}

Monoidal dagger categories have a sound and complete \emph{graphical calculus}, that we briefly recall; for more details, see~\cite{selinger:graphicallanguages}. A morphism $A \sxto{f} B$ is represented as 
$\setlength\morphismheight{3mm}\begin{pic}
  \node[morphism,font=\tiny] (f) at (0,0) {$f$};
  \draw (f.south) to +(0,-.1);
  \draw (f.north) to +(0,.1);
\end{pic}$,
and composition, the tensor product, and the dagger, become:
\[
  \begin{pic}
    \node[morphism] (f) {$g \circ f$};
    \draw (f.south) to +(0,-.65) node[below] {$A$};
    \draw (f.north) to +(0,.65) node[above] {$C$};
  \end{pic}
  = 
  \begin{pic}
    \node[morphism] (g) at (0,.75) {$g\vphantom{f}$};
    \node[morphism] (f) at (0,0) {$f$};
    \draw (f.south) to +(0,-.3) node[below] {$A$};
    \draw (g.south) to node[right] {$B$} (f.north);
    \draw (g.north) to +(0,.3) node[above] {$C$};
  \end{pic}
  \qquad\qquad
  \begin{pic}
    \node[morphism] (f) {$f \otimes g$};
    \draw (f.south) to +(0,-.65) node[below] {$A \otimes C$};
    \draw (f.north) to +(0,.65) node[above] {$B \otimes D$};
  \end{pic}
  = 
  \begin{pic}
    \node[morphism] (f) at (-.4,0) {$f$};
    \node[morphism] (g) at (.4,0) {$g\vphantom{f}$};
    \draw (f.south) to +(0,-.65) node[below] {$A$};
    \draw (f.north) to +(0,.65) node[above] {$B$};
    \draw (g.south) to +(0,-.65) node[below] {$C$};
    \draw (g.north) to +(0,.65) node[above] {$D$};
  \end{pic}
  \qquad\qquad
  \begin{pic}
    \node[morphism] (f) {$f^\dag$};
    \draw (f.south) to +(0,-.65) node[below] {$B$};
    \draw (f.north) to +(0,.65) node[above] {$A$};
  \end{pic}
  =
  \begin{pic}
    \node[morphism,hflip] (f) {$f$};
    \draw (f.south) to +(0,-.65) node[below] {$B$};
    \draw (f.north) to +(0,.65) node[above] {$A$};
  \end{pic}
\]
Notice that the output wire $B \otimes D$ of a morphism $A \sxto{f} B \otimes D$ becomes a pair of wires labelled $B$ and $D$ coming out of the box labelled $f$. Also, the dagger reflects in the horizontal axis, which is why we draw the boxes asymmetrically. Distinguished morphisms are often depicted with special diagrams instead of generic boxes as above. For example, the identity $A \to A$ is just the line
$\begin{pic}
  \draw (0,-.3) to (0,0);
\end{pic}$\;; 
the (identity on) the monoidal unit object $I$ is drawn as the empty picture, and the swap map of symmetric monoidal categories becomes 
$\begin{pic}[scale=.25]
  \draw (0,0) to[out=80,in=-100] (1,1);
  \draw (1,0) to[out=100,in=-80] (0,1);
\end{pic}$.
Soundness and completeness means that any equality between morphisms one can prove algebraically using the axioms of monoidal dagger categories can equivalently and rigorously be proven graphically by isotopies of the graphical diagram.

To model \emph{higher order computation}, we need function types. This is usually done by requiring \emph{closed monoidal categories}, where the functors $- \otimes B$ have right adjoints $B \multimap -$. That is, there is a natural bijective correspondence between morphisms $B \otimes A \sxto{f} C$ and their \emph{curried} version $A \sxto{\Lambda(f)} (B \multimap C)$. In the reversible setting of monoidal dagger categories, this closure operation should cooperate with the dagger: since $B \multimap C$ is the type of computations $B \sxto{f} C$, and those computations can be reversed to $C \sxto{f^\dag} B$, there should be an operation $(B \multimap C) \to (C \multimap B)$ modelling this internally (we will see this in more detail in Section~\ref{sec:monoids}). Therefore we demand that $B \multimap -$ are dagger functors. It follows that they are not just right adjoint to $- \otimes B$, but also left adjoint. Now it is a small step to so-called \emph{compact dagger categories}~\cite{lindner:adjunctions,kellylaplaza:compactcategories}, which we make here for the sake of simplicity.

\begin{definition}
  A \emph{compact dagger category} is a symmetric monoidal dagger category in which every object $A$ has a chosen dual object $A^*$ and a morphism $I \sxto{u} A^* \otimes A$, drawn as 
  $\begin{pic}
    \draw (0,0) to[out=-90,in=-90,looseness=1.25] (.4,0);
  \end{pic}$, 
  satisfying $\big((u^\dag \circ \sigma) \otimes \id\big) \circ (\id \otimes u) = \id$ and its dual:
  \begin{equation}\label{eq:snake}
    \begin{pic}
      \draw (0,0) node[left] {$A$} to[out=90,in=-90] (.5,1);
      \draw (.5,1) to[out=90,in=90,looseness=1.25] (0,1);
      \draw (0,1) to[out=-90,in=90] node[above right=-.5mm] {$\;\;A^*$} (1,.5);
      \draw (1,.5) to[out=-90,in=-90,looseness=1.25] (1.5,.5);
      \draw (1.5,.5) to (1.5,2) node[left] {$A$};
    \end{pic}
    =
    \begin{pic}
      \draw (0,0) node[right] {\phantom{$A$}} to (0,2) node[right] {$A$};
    \end{pic}
    \qquad
    \begin{pic}
      \draw (0,0) node[left] {\phantom{$A^*$}} to (0,2) node[left] {$A^*$};
    \end{pic}
    =
    \begin{pic}
      \draw (1.5,0) node[right] {$A^*$} to[out=90,in=-90] (1,1);
      \draw (1.5,1) to[out=90,in=90,looseness=1.25] (1,1);
      \draw (1.5,1) to[out=-90,in=90] node[above left=-.5mm] {$A\;\;\;$} (.5,.5);
      \draw (0,.5) to[out=-90,in=-90,looseness=1.25] (.5,.5);
      \draw (0,.5) to (0,2) node[right] {$A^*$};
    \end{pic}
  \end{equation}
\end{definition}
\noindent
Compact dagger categories are automatically closed monoidal, with $(B \multimap C) = B^* \otimes C$. Think of dual objects $B^*$ as \emph{input} types, and primal objects $C$ as \emph{output} types. By convention we choose $A^{**}=A$ and $(A \otimes B)^* = B^* \otimes A^*$.

Our previous examples in fact already satisfy this closure property of \emph{higher order computation}: 
$\cat{Rel}$ and $\cat{FHilb}$ are compact dagger categories as follows. In $\cat{Rel}$ we can take $A^*=A$ and $u = \{(*,(a,a)) \mid a \in A\}$ for $I=\{*\}$. In $\cat{FHilb}$ we can take $H^*$ to be the dual Hilbert space of $H$; if $H$ has an orthonormal basis $\{e_1,\ldots,e_n\}$, then $H^*$ has an orthonormal basis $\{e_1^*,\ldots,e_n^*\}$, and we can take $u(1) = \sum_{i=1}^n e_i^* \otimes e_i$.
There is also a \emph{free compact dagger category} on a given (dagger) category $\cat{C}$~\cite{abramsky:free}.

Let us conclude this preparatory section by contrasting reversible computing and \emph{invertible computing}.
A \emph{groupoid} is a category where any morphism is invertible; it is always a dagger category with $f^\dag = f^{-1}$.
Any symmetric monoidal closed groupoid $\cat{G}$ is a so-called compact category with $A^* = (A \multimap I)$, as follows.
Closure gives isomorphisms $(A \multimap B) \otimes A \sxto{\mathrm{ev}} B$ for all objects $A$ and $B$; in particular, $I \cong A^* \otimes A$. 
The morphisms $\Lambda(\mathrm{ev})$ are isomorphisms $A \cong A^{**}$, making $\cat{G}$ into a so-called $*$-autonomous category~\cite{barrwells:computer}. 
Because $\cat{G}$ is symmetric monoidal, there are isomorphisms $A^* \otimes B^* \sxto{\Lambda(\mathrm{ev} \otimes \mathrm{ev})} (A \otimes B)^*$, making $\cat{G}$ a compact category.
However, this is not a compact dagger category unless all swap maps $\sigma$ are identities.

\section{Frobenius monoids}\label{sec:monoids}

This section considers monoids in monoidal dagger categories. We will see that, in the higher order setting of closed monoidal categories, our rule of thumb that everything should cooperate with the dagger means considering \emph{Frobenius monoids}.

\begin{definition}\label{def:monoid}
  A \emph{monoid} in a monoidal category is an object $A$ with morphisms $\tinymult \colon A \otimes A \to A$ and $\tinyunit \colon I \to A$, satisfying:
  \[
    \begin{pic}[scale=.4]
      \node[dot] (t) at (0,1) {};
      \node[dot] (b) at (1,0) {};
      \draw (t) to +(0,1);
      \draw (t) to[out=0,in=90] (b);
      \draw (t) to[out=180,in=90] (-1,0) to (-1,-1);
      \draw (b) to[out=180,in=90] (0,-1);
      \draw (b) to[out=0,in=90] (2,-1);
    \end{pic}
    =
    \begin{pic}[yscale=.4,xscale=-.4]
      \node[dot] (t) at (0,1) {};
      \node[dot] (b) at (1,0) {};
      \draw (t) to +(0,1);
      \draw (t) to[out=0,in=90] (b);
      \draw (t) to[out=180,in=90] (-1,0) to (-1,-1);
      \draw (b) to[out=180,in=90] (0,-1);
      \draw (b) to[out=0,in=90] (2,-1);
    \end{pic}
  \qquad
  \begin{pic}[scale=.4]
    \node[dot] (d) {};
    \draw (d) to +(0,1);
    \draw (d) to[out=0,in=90] +(1,-1) to +(0,-1);
    \draw (d) to[out=180,in=90] +(-1,-1) node[dot] {};
  \end{pic}
  =
  \begin{pic}[scale=.4]
    \draw (0,0) to (0,3);
  \end{pic}
  =
  \begin{pic}[yscale=.4,xscale=-.4]
    \node[dot] (d) {};
    \draw (d) to +(0,1);
    \draw (d) to[out=0,in=90] +(1,-1) to +(0,-1);
    \draw (d) to[out=180,in=90] +(-1,-1) node[dot] {};
  \end{pic}
  \]
  It is \emph{commutative} when $\tinymult = \tinymult \circ \sigma$.
  A \emph{Frobenius monoid} is a monoid in a monoidal dagger category satisfying~\eqref{eq:frobeniuslaw}.
  It is \emph{special} when $\tinymult \circ \big(\tinymult\big)^\dag = \id[A]$.
\end{definition}

A \emph{comonoid} in $\cat{C}$ is a monoid in $\cat{C}\op$. The Frobenius law~\eqref{eq:frobeniuslaw} makes sense for pairs of a monoid and comonoid on the same object, and most of Section~\ref{sec:monads} holds in that generality. 
Each side of the Frobenius law~\eqref{eq:frobeniuslaw} equals $(\tinymult)^\dag \circ \tinymult$; one of these equations is equivalent to~\eqref{eq:frobeniuslaw}.
It is mostly motivated by observing that Frobenius monoids in specific categories are appropriate well-known mathematical structures.

\begin{example}\label{ex:cstaralgebras}
  Frobenius monoids in $\cat{FHilb}$ correspond to finite-dimensional \emph{C*-algebras}~\cite[Theorem~4.6]{vicary:quantumalgebras}.
  These play a major role in quantum computing~\cite{keyl:quantuminformation}, but also as semantics for labelled Markov processes with bisimulations~\cite{misloveetal:markov,sahebdjahromi:cpos,kozen:probabilistic,moshierpetrisan:duality} and as operational semantics of probabilistic languages~\cite{dipierrohankinwiklicky:approximate,dipierrowiklicky:operator}. 
  Commutative Frobenius monoids in $\cat{FHilb}$ therefore correspond to orthonormal bases \todo{when special}~\cite{coeckepavlovicvicary:bases}.
\end{example}

\begin{example}\label{ex:groupoids}
  Frobenius monoids in $\cat{Rel}$ correspond to (small) \emph{groupoids}~\cite{heunencontrerascattaneo:groupoids,pavlovic:frobrel}, which are important to invertible computing.
\end{example}

\begin{example}
  In a compact dagger category, $A^* \otimes A$ is a Frobenius monoid with $\tinyunit = u$, and $\tinymult$ being the \emph{pair of pants}:
  
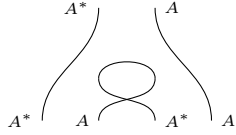
\[
    \begin{pic}[scale=.75,font=\tiny]
      \draw (0,0) node[left] {$A^*$} to[out=90,in=-90] (1,2) node[left] {$A^*$};
      \draw (3,0) node[right] {$A$} to[out=90,in=-90] (2,2) node[right] {$A$};
      \draw (1,0) node[left] {$A$} to[out=90,in=-90] (2,.75);
      \draw (2,0) node[right] {$A^*$} to[out=90,in=-90] (1,.75);
      \draw (2,.75) to[out=90,in=90] (1,.75);
    \end{pic}
  \]
  This is precisely the monoid $A \multimap A$ of computations $A \to A$ under composition.
\end{example}

Pair of pants are universal, as the following generalization of Cayley's theorem shows.
A \emph{monoid homomorphism} $f$ satisfies $\tinyunit = f \circ \tinyunit$ and $f \circ \tinymult = \tinymult \circ (f \otimes f)$.

\begin{lemma}
  Any monoid $(A,\tinymult,\tinyunit)$ in a compact category allows a monic monoid homomorphism $R$ into $A^* \otimes A$.
\end{lemma}
\begin{proof}
  The following is a monoid homomorphism by~\eqref{eq:frobeniuslaw}:
  \begin{equation}\label{eq:cayleyembedding}
    \begin{pic}
      \node[morphism] (R) {$R$};
      \draw (R.south) to +(0,-.3) node[right] {$A$};
      \draw (R.north east) to +(0,.3) node[right] {$A$};
      \draw (R.north west) to +(0,.3) node[left] {$A^*$};
    \end{pic}
    \quad = 
    \begin{pic}
      \node[dot] (d) at (0,0) {};
      \draw (d.north) to (0,.3) node[right] {$A$};
      \draw (d.east) to[out=0, in=90] (.3,-.25) to (.3,-.8) node[left] {$A$};
      \draw (d.west) to[out=180, in=90] (-.3,-.25);
      \draw (-.8,-.25) to[out=-90,in=-90,looseness=1.5] (-.3,-.25);
      \draw (-.8,-.25) to (-.8,.3) node[left] {$A^*$};
    \end{pic}
  \end{equation}
  It is monic because it has a left inverse $\big((\tinyunit)^\dag \otimes \id\big) \circ R$.
\end{proof}

We will prove that the Cayley embedding of the previous lemma respects daggers precisely when the monoid is a Frobenius monoid.
To make precise what it means to respect daggers, we need to internalize the operation $f \mapsto f^\dag$ from $A \sxto{f} A$ to the monoid $A \multimap A$. But the former might not be a well-defined morphism; for example, in $\cat{FHilb}$, taking conjugate transpose matrices is \emph{anti-}linear, not linear, and hence a morphism $(A \multimap A) \to (A \multimap A)^*$ rather than an endomorphism. In a compact category, this is modeled by
\[
  \begin{pic}
    \node[morphism] (f) {$f_*$};
    \draw (f.north) to +(0,.5) node[right] {$B^*$};
    \draw (f.south) to +(0,-.5) node[right] {$A^*$};
  \end{pic}
  \quad := \quad
  \quad = \quad
  \begin{pic}
    \node[morphism,hflip] (f) {$f$};
    \draw (f.north) to[out=90,in=-90] +(.75,.4) to[out=90,in=90,looseness=.75] +(-.75,0) to[out=-90,in=90] +(.75,-.4) to +(0,-.9) node[left] {$A^*$};
    \draw (f.south) to[out=-90,in=-90] +(-.75,0) to +(0,.9) node[right] {$B^*$};
  \end{pic}
\]
for $A \sxto{f} B$.
The operation $f \mapsto f^\dag$ additionally is contravariant: $(g \circ f)^\dag = f^\dag \circ g^\dag$. So for it to be a monoid homomorphism the codomain has to have opposite multiplication as the domain.

\begin{lemma}
  If $(A,\tinymult,\tinyunit)$ is a monoid in a compact category, then so is $(A^*,\tinymult_*,\tinyunit_*)$, called the \emph{opposite monoid}.
\end{lemma}
\begin{proof}
  The functor $f \mapsto f_*$ is (strong) monoidal.
\end{proof}

\begin{definition}
  A monoid $(A,\tinymult,\tinyunit)$ in a compact dagger category is an \emph{involutive monoid} when it is equipped with an \emph{involution}: a monoid homomorphism $A \sxto{i} A^*$ satisfying $i_* \circ i = \id$.
  A \emph{homomorphism of involutive monoids} is a monoid homomorphism $A \sxto{f} B$ satisfying $i \circ f = f_* \circ i$.
\end{definition}

Note that there is a canonical choice of involution:
\begin{equation}\label{eq:canonicalinvolution}
  \begin{pic}
    \node[dot] (d) at (0,0) {};
    \draw (d.north) to +(0,.2) node[dot] {};
    \draw (d.east) to[out=0,in=90] +(.3,-.4) to +(0,-.4) node[left] {$A$};
    \draw (d.west) to[out=180,in=90] +(-.3,-.4) to[out=-90,in=-90] +(-.7,0) to +(0,1) node[right] {$A^*$};
  \end{pic}
\end{equation}
For the groupoids of Example~\ref{ex:groupoids}, it is $g \mapsto g^{-1}$. For the C*-algebras of Example~\ref{ex:cstaralgebras}, it is $a \mapsto a^*$. The following theorem justifies the Frobenius law from first principles, generalizing~\cite[Theorem~4.3]{pavlovic:abstraction} noncommutatively.

\begin{theorem}
  A monoid in a compact dagger category is a Frobenius monoid if and only if~\eqref{eq:canonicalinvolution} makes it involutive and~\eqref{eq:cayleyembedding} a homomorphism of involutive monoids.
\end{theorem}
\begin{proof}
  Write $(A,\tinymult,\tinyunit)$ for the monoid, and $i$ for~\eqref{eq:canonicalinvolution}.
  If $A$ is a Frobenius monoid, it follows from~\eqref{eq:frobeniuslaw} that $i$ is indeed an involution.
  Observe that the involution on $A^* \otimes A$ is the identity because of our convention $(A \otimes B)^* = B^* \otimes A^*$.
  So~\eqref{eq:cayleyembedding} preserves involutions when $R_* \circ i = R$:
  \[
\todo{
    \begin{pic}
      \node[morphism] (i) at (0,.1) {$i$};
      \node[morphism] (R) at (0,.9) {$R_*$};
      \draw (R.south) to node[left] {$A^*$} (i.north);
      \draw (i.south) to +(0,-.3) node[left] {$A$};
      \draw ([xshift=1mm]R.north east) to +(0,.3) node[right] {$A$};
      \draw ([xshift=-1mm]R.north west) to +(0,.3) node[left] {$A^*$};
    \end{pic}
}
    = \hspace*{-2mm}
    \begin{pic}
      \node[dot] (l) at (0,0) {};
      \node[dot] (r) at (.75,.6) {};
      \draw (r.north) to +(0,.25) node[dot] {};
      \draw (r.east) to[out=0,in=90] +(.3,-.3) to +(0,-.75) node[left] {$A$};
      \draw (r.west) to[out=180,in=0] (l.east);
      \draw (l.west) to[out=180,in=-90] +(-.2,.3) to +(0,1.075) node[right] {$A$};
      \draw (l.south) to[out=-90,in=-90,looseness=1.5] +(-.6,0) to +(0,1.475) node[left] {$A^*$};
    \end{pic}
    \;\; = \;\;
    \begin{pic}
      \node[dot] (d) at (0,0) {};
      \draw (d.north) to (0,.8) node[right] {$A$};
      \draw (d.east) to[out=0, in=90] (.3,-.25) to (.3,-1) node[left] {$A$};
      \draw (d.west) to[out=180, in=90] (-.3,-.25);
      \draw (-.8,-.25) to[out=-90,in=-90,looseness=1.5] (-.3,-.25);
      \draw (-.8,-.25) to (-.8,.8) node[right] {$A^*$};
    \end{pic}
    \; = \hspace*{-2mm}
    \begin{pic}
      \node[morphism] (R) at (0,.5) {$R$};
      \draw ([xshift=1mm]R.north east) to +(0,.7) node[right] {$A$};
      \draw ([xshift=-1mm]R.north west) to +(0,.7) node[left] {$A^*$};
      \draw (R.south) to +(0,-.7) node[right] {$A$};
    \end{pic}
  \]

  Conversely, assuming $R_* \circ i = R$:
  \begin{equation}\tag{$*$}
    \begin{pic}
      \node[dot] (d) {};
      \draw (d.north) to +(0,.6);
      \draw (d.east) to[out=0,in=90] +(.3,-.3) to +(0,-.7);
      \draw (d.west) to[out=180,in=90] +(-.3,-.3) to +(0,-.7);
    \end{pic}
    \;=\;
    \begin{pic}
      \node[morphism] (R) at (0,0) {$R$};
      \draw (R.south) to +(0,-.6);
      \draw ([xshift=1mm]R.north east) to +(0,.7);
      \draw ([xshift=-1mm]R.north west) 
        to[out=90,in=-90] +(-.5,.4)
        to[out=90,in=90,looseness=.8] +(.5,0)
        to[out=-90,in=90] +(-.5,-.4)
        to +(0,-1);
    \end{pic}
    \;=\;
\todo{
    \begin{pic}
      \node[morphism] (R) at (0,.5) {$R_*$};
      \node[morphism] (i) at (0,0) {$i$};
      \draw (R.south) to (i.north);
      \draw (i.south) to +(0,-.2);
      \draw ([xshift=1mm]R.north east) to +(0,.6);
      \draw ([xshift=-1mm]R.north west) 
        to[out=90,in=-90] +(-.5,.4)
        to[out=90,in=90,looseness=.8] +(.5,0)
        to[out=-90,in=90] +(-.5,-.4)
        to +(0,-1.1);
    \end{pic}
}
    \;=\;
    \begin{pic}
      \node[dot] (l) at (-.4,-.3) {};
      \node[dot] (r) at (.4,.3) {};
      \draw (l.east) to[out=0,in=180] (r.west);
      \draw (r.north) to +(0,.2) node[dot] {};
      \draw (r.east) to[out=0,in=90] +(.3,-.3) to +(0,-.9);
      \draw (l.south) to +(0,-.5);
      \draw (l.west) to[out=180,in=-90] +(-.3,.3) to +(0,.8);
    \end{pic}
  \end{equation}
  Hence, by associativity:
  \[
    \begin{pic}
      \node[dot] (l) at (-.4,-.3) {};
      \node[dot] (r) at (.4,.3) {};
      \draw (l.east) to[out=0,in=180] (r.west);
      \draw (r.north) to +(0,.4);
      \draw (r.east) to[out=0,in=90] +(.3,-.3) to +(0,-.9);
      \draw (l.south) to +(0,-.5);
      \draw (l.west) to[out=180,in=-90] +(-.3,.3) to +(0,.8);
    \end{pic}
    \;\stackrel{(*)}{=}\;
    \begin{pic}
      \node[dot] (b) at (-.7,-.6) {};
      \node[dot] (l) at (-.3,-.2) {};
      \node[dot] (r) at (.4,.3) {};
      \draw (l.east) to[out=0,in=180] (r.west);
      \draw (r.north) to +(0,.2) node[dot] {};
      \draw (r.east) to[out=0,in=90] +(.3,-.3) to +(0,-.9);
      \draw (l.west) to[out=180,in=-90] +(-.3,.3) to +(0,.7);
      \draw (l.south) to[out=-90,in=0] (b.east);
      \draw (b.south) to +(0,-.2);
      \draw (b.west) to[out=180,in=-90] +(-.3,.3) to +(0,1.1);
    \end{pic}
    \;=
    \begin{pic}
      \node[dot] (t) at (-.7,.2) {};
      \node[dot] (l) at (-.4,-.3) {};
      \node[dot] (r) at (.4,.3) {};
      \draw (l.east) to[out=0,in=180] (r.west);
      \draw (r.north) to +(0,.2) node[dot] {};
      \draw (r.east) to[out=0,in=90] +(.3,-.3) to +(0,-.9);
      \draw (l.south) to +(0,-.5);
      \draw (l.west) to[out=180,in=-90] (t.south);
      \draw (t.east) to[out=0,in=-90] +(.3,.3) to +(0,.3);
      \draw (t.west) to[out=180,in=-90] +(-.3,.3) to +(0,.3);
    \end{pic}
    \;\stackrel{(*)}{=}\;
    \begin{pic}
      \node[dot] (t) at (0,.5) {};
      \node[dot] (b) at (0,0) {};
      \draw (t.south) to (b.north);
      \draw (t.east) to[out=0,in=-90] +(.3,.3) to +(0,.3);
      \draw (t.west) to[out=180,in=-90] +(-.3,.3) to +(0,.3);
      \draw (b.east) to[out=0,in=90] +(.3,-.3) to +(0,-.3);
      \draw (b.west) to[out=180,in=90] +(-.3,-.3) to +(0,-.3);
    \end{pic}
  \]
  But this is equivalent to~\eqref{eq:frobeniuslaw}.
\end{proof}


\section{Frobenius monads}\label{sec:monads}

A monad is a functor $\cat{C} \sxto{T} \cat{C}$ with natural transformations $T(T(A)) \sxto{\mu_A} T(A)$ and $A \sxto{\eta_A} T(A)$ satisfying certain laws. It is well-known that monads are precisely monoids in categories of functors $\cat{C} \to \cat{C}$: Definition~\ref{def:monoid} unfolds to the monad laws
\begin{align*}
  \mu_A \circ T(\mu_A) & = \mu_A \circ \mu_{T(A)}, \\
  \mu_A \circ T(\eta_A) = \;& \id[T(A)] = \mu_A \circ \eta_{T(A)}.
\end{align*}
There is a dual notion of a comonad. Daggers make any monoid (monad) give rise to a comonoid (comonad). Thus the Frobenius law~\eqref{eq:frobeniuslaw} lifts to monads as follows. 

\begin{definition}
  A \emph{Frobenius monad} on a dagger category $\cat{C}$ is a Frobenius monoid in $[\cat{C},\cat{C}]_\dag$;
  explicitly, a monad $(T,\mu,\eta)$ on $\cat{C}$ with $T(f^\dag)=T(f)^\dag$ and 
  \[
    T(\mu_A) \circ \mu^\dag_{T(A)} = \mu_{T(A)} \circ T(\mu_A^\dag).
  \]
  It is \emph{special} when $\mu_A \circ \mu^\dag_A = \id[T(A)]$.
\end{definition}

Frobenius monads have been studied before by Street~\cite{street:ambidextrous,lauda:ambidextrous}. His definition does not take daggers into account, and concerns a monad rather than a monad-comonad pair. 
However, the natural generalization of the above definition to (non-dagger) monad-comonad pairs results in an equivalent notion to the one studied by Street.
The primary example of a Frobenius monad is taking tensor products with a Frobenius monad.

\begin{example}\label{example:tensor}
   If $(B,\tinymult,\tinyunit)$ is a Frobenius monoid in a monoidal dagger category $\cat{C}$, then the functor $\cat{C} \sxto{- \otimes B} \cat{C}$, given by $A \mapsto A \otimes B$ and $f \mapsto f \otimes \id$, is a Frobenius monad on $\cat{C}$ with:
   \[
     \mu_A = \begin{pic}
       \node[dot] (d) at (0,0) {};
       \draw (d.north) to (0,.5) node[right] {$B$};
       \draw (d.east) to[out=0,in=90] +(.3,-.3) to +(0,-.2) node[right] {$B$};
       \draw (d.west) to[out=180,in=90] +(-.3,-.3) to +(0,-.2) node[right] {$B$};
       \draw (-.8,-.5) node[left] {$A$} to (-.8,.5) node[left] {$A$};
     \end{pic} 
     \qquad
     \eta_A = \begin{pic}
       \node[dot] (d) at (0,0) {};
       \draw (d.north) to (0,.5) node[right] {$B$};
       \draw (-.4,-.5) node[left] {$A$} to (-.4,.5) node[left] {$A$};
     \end{pic} 
   \]
\end{example}
\begin{proof}
  The Frobenius monad law simply comes down to the Frobenius monoid law:
  \[
    T\mu \circ \mu^\dag_T
    =
    \begin{pic}[scale=.75]
      \draw (0,0) node[below] {$B$} to (0,1) to[out=90,in=180] (.5,1.5) to (.5,2) node[above] {$B$};
      \draw (.5,1.5) to[out=0,in=90] (1,1) to[out=-90,in=180] (1.5,.5) to (1.5,0) node[below] {$B$};
      \draw (1.5,.5) to[out=0,in=-90] (2,1) to (2,2) node[above] {$B$};
      \node[dot] at (.5,1.5) {};
      \node[dot] at (1.5,.5) {};
      \draw (-.5,0) node[below] {$A$} to (-.5,2) node[above] {$A$};
    \end{pic}
    =
    \begin{pic}[yscale=.75,xscale=-.75]
      \draw (0,0) node[below] {$B$} to (0,1) to[out=90,in=180] (.5,1.5) to (.5,2) node[above] {$B$};
      \draw (.5,1.5) to[out=0,in=90] (1,1) to[out=-90,in=180] (1.5,.5) to (1.5,0) node[below] {$B$};
      \draw (1.5,.5) to[out=0,in=-90] (2,1) to (2,2) node[above] {$B$};
      \node[dot] at (.5,1.5) {};
      \node[dot] at (1.5,.5) {};
      \draw (2.5,0) node[below] {$A$} to (2.5,2) node[above] {$A$};
    \end{pic}
    =
    \mu_T \circ T\mu^\dag
  \]
  The monad laws become the monoid laws. Taking $A=I$, we thus see that $- \otimes B$ is a Frobenius monad if and only if $B$ is a Frobenius monoid.
\end{proof}

This section characterizes Frobenius monads of this form. There are, however, also other Frobenius monads, as in the following example.

\begin{example}
  Consider the monoid $\cat{Rel}(\mathbb{N},\mathbb{N})$ of all relations $\mathbb{N}\to\mathbb{N}$ as a single-object category. The following define a Frobenius monad on this category: 
  \begin{align*}
    T(R) = & \{(2m,2n)\mid (m,n)\in R\} \\
         & \cup\{(2m+1,2n+1)\mid (m,n)\in R\} \\ 
    \eta =& \{(2n,2n+1) \mid n \in \mathbb{N}\}  \\
    \mu =& \{(4n,2n) \mid n \in \mathbb{N}\} 
         \cup \{ (4n+3,2n+1) \mid n \in \mathbb{N}\} 
  \end{align*}
\end{example}  

The functor $-\otimes B$ comes with a natural transformation $\alpha_{-,-,B}$, making it a strong functor. This natural transformation respects the monoid structure on $B$. Before recording some folklore results, we first define what this means for monads.

\begin{definition}
  A functor $F$ between monoidal categories is \emph{strong} when it is equipped with a natural transformation
  $A \otimes F(B) \sxto{\st_{A,B}} F(A\otimes B)$ 
  satisfying $\st \circ \alpha = F(\alpha) \circ \st \circ (\id \otimes \st)$ and $F(\lambda) \circ \st = \lambda$.
  A \emph{morphism of strong functors} is a natural transformation $F\sxto{\beta} G$
  satisfying $\beta \circ \st = \st \circ (\id \otimes \beta)$
  A \emph{strong monad} is a monad $(T,\mu,\eta)$ that is a strong functor
  satisfying $\st \circ (\id \otimes \mu) = \mu \circ T(\st) \circ \st$ and $\st \circ (\id \otimes \eta) = \eta$.
  A \emph{morphism of strong monads} is a natural transformation, which is a morphism of the underlying monads and the underlying strong functors.
\end{definition} 
 

\begin{proposition}\label{prop:monoids}
  Let $\cat{C}$ be a monoidal category. The operations $B \mapsto - \otimes B$ and $T\mapsto T(I)$ define an adjunction between monoids in $\cat{C}$ and strong monads on $\cat{C}$, with $B \mapsto -\otimes B$ being the left adjoint.
\end{proposition}
\begin{proof}
  See~\cite{wolff:monads}.  
  The unit of the adjunction is $I \otimes B \sxto{\lambda_B} B$.
  The counit is determined by $A \otimes T(I) \sxto{T(\rho) \circ\, \st} T(A)$.
\end{proof}

In the case of symmetric monoidal categories, there is also a notion of commutativity for strong monads~\cite{kock:strong,jacobs:weakening}. 
Given a strong monad $T$, one can define a natural transformation $T(A) \otimes B \sxto{\st'_{A,B}} T(A \otimes B)$ by $T(\sigma_{B,A})\circ st_{B,A}\circ \sigma_{T(A),B}$, and
\begin{align*} 
  dst_{A,B} & := \mu_{A\otimes B}\circ T(st'_{A,B})\circ st_{T(A),B} \\ 
  dst_{A,B}' & :=\mu_{A\otimes B}\circ T(st_{A,B})\circ st'_{A,T(B)}
\end{align*} 
A strong monad is \emph{commutative} when these coincide.
Proposition~\ref{prop:monoids} restricts to an adjunction between commutative monoids and commutative monads~\cite{wolff:monads}.



\begin{definition}
  A \emph{costrong functor} $\cat{C} \sxto{F} \cat{D}$ between monoidal categories is a functor that is strong when considered as a functor $\cat{C}\op \to \cat{D}\op$. Explicitly, it has a natural transformation $F(A \otimes B) \sxto{\cst_{A,B}} A \otimes F(B)$ satisfying
   $F(\lambda) =\lambda \circ \cst_{I,A}$ and 
   $\cst_{A\otimes B,C}\circ F(\alpha) =\alpha \circ (\id \otimes \cst_{B,C})\circ \cst_{A,B\otimes C}$.
  A \emph{morphism of costrong functors} is a natural transformation $F\sxto{\beta} G$ satisfying
  $
    \cst \circ \beta = (\id \otimes \beta) \circ \cst
  $.
  A \emph{costrong comonad} is a comonad $(T,\delta,\varepsilon)$ that is a costrong functor, such that
    $(\id \otimes \varepsilon) \circ \cst = \varepsilon$ and
    $\cst \circ T(\cst) \circ \delta = (\id \otimes \delta) \circ \cst$.
  A \emph{morphism of costrong comonads} is a natural transformation, which is a morphism of the underlying comonads and the underlying costrong functors.
\end{definition}


\begin{corollary}\label{cor:comonoids}
  Let $\cat{C}$ be a monoidal category. The operations $B \mapsto -\otimes B$ and $T\mapsto T(I)$ form an adjunction between comonoids in $\cat{C}$ and costrong comonads on $\cat{C}$, but this time $B \mapsto -\otimes B$ is the right adjoint. 
  \qed
\end{corollary}

In our reversible setting of dagger categories, any strong monad $T$ is automatically a costrong comonad under
$\cst = \st^\dag$, $\delta = \mu^\dag$, and $\varepsilon = \eta^\dag$. According to our motto that everything in sight should cooperate with the dagger, the reverse $\cst$ of $\st$ should in fact be its inverse, leading to the following definition.

\begin{definition}
  A \emph{strong Frobenius monad} on a monoidal dagger category $\cat{C}$ is a Frobenius monad $(T,\mu,\eta)$ that is simultaneously a strong monad, such that each $\st_A$ is unitary.
  A \emph{morphism of strong Frobenius monads} is just a morphism of the underlying strong monads.
\end{definition}

The following theorem promotes the adjunction of Proposition~\ref{prop:monoids} and Corollary~\ref{cor:comonoids} into an equivalence in the dagger setting. It generalizes \cite[Theorem~4.5]{pavlovic:abstraction} noncommutatively. It also generalizes the classic Eilenberg--Watts theorem, that characterizes certain endofunctors on abelian categories as being of the form $- \otimes B$ for a monoid $B$, to monoidal dagger categories; note that there are monoidal dagger categories that are not abelian, such as $\cat{Rel}$ and $\cat{Hilb}$~\cite[Appendix~A]{heunen:embedding}.

\begin{theorem}\label{thm:strong}
  Let $\cat{C}$ be a monoidal dagger category. The operations $B \mapsto -\otimes B$ and $T\mapsto T(I)$ define an equivalence between Frobenius monoids in $\cat{C}$ and strong Frobenius monads on $\cat{C}$. 
\end{theorem} 
\begin{proof}
  We already saw in Example~\ref{example:tensor} that $B\mapsto -\otimes B$ preserves the Frobenius law. We prove that $T\mapsto T(I)$ preserves the Frobenius law, too, in Lemma~\ref{lem:froblaw} in the Appendix. It remains to prove that they form an equivalence. Clearly the unit of the adjunction, $I \otimes B \sxto{\lambda_B} B$, is a natural isomorphism. 
  To prove that the counit $A \otimes T(I) \sxto{T(\rho) \circ \st} T(A)$ is also a natural isomorphism, notice that by definition it is a morphism of strong monads. In Lemma~\ref{lem:strictmorphism} in the Appendix we prove that it is also a morphism of comonads. But homomorphisms of Frobenius monoids must be isomorphisms by Lemma~\ref{lem:strictmorphismsareiso}.
\end{proof}

The previous theorem restricts to an equivalence between commutative/special Frobenius monoids and commutative/special strong Frobenius monads (see Corollary~\ref{cor:special} in the Appendix).


One might think it too strong to require $\st$ to be unitary. The following counterexample shows that Theorem~\ref{thm:strong} would fail if we abandoned that requirement.

\begin{example}
  Let's call a Frobenius monad \emph{rather strong} when it is simultaneously a strong monad. The operations of Theorem~\ref{thm:strong} do not form an adjunction between Frobenius monoids and rather strong Frobenius monads, because the counit of the adjunction would not be a well-defined morphism. To produce a counterexample where the counit does not preserve comultiplication comes down to finding a rather strong Frobenius monad with $T(\eta_A) \circ \eta_A \neq \mu_A^\dag \circ \eta_A$ for some $A$. This is the case when $T$ is $- \otimes B$ for a Frobenius monoid $B$ with $\tinyunit \otimes \tinyunit \neq (\tinymult)^\dag \circ \tinyunit$. Such Frobenius monoids certainly exist: if $G$ is any nontrivial group, regarded as a Frobenius monoid in $\cat{Rel}$ via Example~\ref{ex:groupoids}, then $\tinyunit \otimes \tinyunit$ is the relation $\{(*,(1,1))\}$, but $(\tinymult)^\dag \circ \tinyunit = \{(*,(g,g^{-1})) \mid g \in G\}$.
\end{example}

\section{Kleisli algebras}\label{sec:kleisli}

One of the standard categorical constructions when given a monad $T$ is to consider the category $\cat{C}_T$ of its Kleisli algebras. In monadic programming, this category gives semantics for computations with effects modeled by $T$, whereas the base category $\cat{C}$ only gives semantics for pure computations~\cite{jacobsheunenhasuo:arrows}. In this section we show that if $T$ is a Frobenius monad, then $\cat{C}_T$ is a dagger category. In fact we also show the converse, under a natural condition about cooperation with daggers. Thus effects modeled by a monad can be added without leaving the setting of reversible computations precisely when the monad is a Frobenius monad.

\begin{definition}
  If $\cat{C} \sxto{T} \cat{C}$ is a monad, its \emph{Kleisli category} $\cat{C}_T$ is defined as follows. Objects are the same as in $\cat{C}$. A morphism $A \to B$ in $\cat{C}_T$ is a morphism $A \sxto{f} T(B)$ in $\cat{C}$. Identities are given by $\eta$, and composition of $g$ and $f$ in $\cat{C}_T$ is given by $\mu \circ T(g) \circ f$.

  There is a forgetful functor $\cat{C}_T \to \cat{C}$ given by $A \mapsto T(A)$ on objects and $f \mapsto \mu \circ T(f)$ on morphisms. It has a left adjoint $\cat{C} \to \cat{C}_T$ given by $A \mapsto A$ on objects and $f \mapsto \eta \circ f$ on morphisms.
\end{definition}

\noindent
We now show that for Frobenius monads the Kleisli construction preserves daggers.

\begin{lemma}\label{lem:kleislidagger}
  If $T$ is a Frobenius monad on a dagger category $\cat{C}$, then $\cat{C}_T$ carries a dagger that commutes with the canonical functors $\cat{C}_T \to \cat{C}$ and $\cat{C}\to\cat{C}_T$.
\end{lemma}
\begin{proof} 
  A straightforward calculation establishes that
  \[
    \big(A \sxto{f} T(B)\big)
    \;\mapsto\;
    \big(B \sxto{\eta} T(B) \sxto{\mu^\dag} T^2(B) \sxto{T(f^\dag)} T(A)\big)
  \]
  is a dagger on $\cat{C}_T$ commuting with the canonical functors $\cat{C}\to \cat{C}_T$ and $\cat{C}_T\to \cat{C}$.
\end{proof}

The following theorem proves a converse of the previous lemma, under the natural condition that the ``reverse identity morphisms'' of the Kleisli category equal their own dagger. This gives another characterization of Frobenius monads, in terms of reversibility of their effectful computations.

\begin{theorem} 
   A monad $T$ on a dagger category $\cat{C}$ is a Frobenius monad if and only if $\cat{C}_T$ has a dagger such that:
   \begin{itemize}
     \item the functors $\cat{C} \to \cat{C}_T$ and $\cat{C}_T \to \cat{C}$ are dagger functors;
     \item the morphisms $\mu_A^\dag \colon T(A) \to T^2(A)$ of $\cat{C}$ are self-adjoint when regarded as morphisms $T(A) \to T(A)$ of $\cat{C}_T$.
   \end{itemize}
 \end{theorem} 
\begin{proof}
  One direction follows from Lemma~\ref{lem:kleislidagger} and the observation that with that dagger the morphism $\mu^\dag_A \colon T(A)\to T^2(A)$ is self-adjoint in $\cat{C}_T$. For the other direction, we wish to show that the following diagram commutes for arbitrary $A$.
  \[\begin{tikzpicture}
  \matrix (m) [matrix of math nodes,row sep=1.5em,column sep=4em,minimum width=2em]
    {
     T^2(A) & T^3(A) \\
     T^3(A) & T^2(A) \\};
    \path[->]
    (m-1-1) edge node [left] {$\mu^\dag_{T(A)}$} (m-2-1)
            edge node [above] {$T(\mu^\dag_A)$} (m-1-2)
    (m-2-1) edge node [below] {$T(\mu_A)$} (m-2-2)
    (m-1-2) edge node [right] {$\mu_{T(A)}$} (m-2-2);
  \end{tikzpicture}\]
  Write $\cat{C} \sxto{F} \cat{C}_T$ and $\cat{C}_T \sxto{G} \cat{C}$ for the canonical functors.
  Note that if we consider $\id[T^2(A)]$ and $\eta_{T(A)}\circ\mu_A$ as morphisms of $\cat{C}_T$, then we have $G(\id[T^2(A)])=\mu_{T(A)}$, $G(\eta_{T(A)}\circ\mu_A)=T(\mu_A)$, and $F(\mu_A)=\eta_{T(A)}\circ\mu_A$.
  As $G$ is a dagger functor, we have found preimages of all the morphisms in the diagram. More explicitly, we know 
  that 
  \begin{align*}
    G\big(\id[T^2(A)]\circ F(\mu_A^\dag)\big) & = \mu_{T(A)}\circ T(\mu^\dag_A), \\
    G\big(F(\mu_A)\circ \id[T^2(A)]^\dag\big) & = T(\mu_A)\circ \mu^\dag_{T(A)}. 
  \end{align*}
  Hence it suffices to show $\id[T^2(A)]\circ F(\mu_A^\dag) = F(\mu_A) \circ \id[T^2(A)]^\dag$. As the left hand side is the dagger of the right hand side and $\mu^\dag$ is self-adjoint in $\cat{C}_T$, it suffices to show that either equals $\mu^\dag$. The following calculation does this for the left-hand side:
  \begin{align*} 
    \id[T^2(A)] \circ F(\mu_A^\dag)
    & = \mu_{T(A)} \circ T(\id[T^2(A)]) \circ \eta_{T^2(A)} \circ \mu^\dag_A \\
    & = \mu_{T(A)} \circ \eta_{T^2(A)} \circ \mu^\dag_A 
    = \mu^\dag_A
  \end{align*} 
  This completes the proof.
\end{proof}

Kleisli categories of commutative monads on symmetric monoidal categories are again symmetric monoidal~\cite{day:kleisli}. This extends to the reversible setting.

\begin{theorem}
  If $T$ is a commutative strong Frobenius monad on a symmetric monoidal dagger category $\cat{C}$, then $\cat{C}_T$ is a symmetric monoidal dagger category.
\end{theorem}
\begin{proof} 
  The monoidal structure on $\cat{C}_T$ is given by $A\otimes_T B=A\otimes B$ on objects and by $f\otimes_T g=\dst\circ (f\otimes g)$ on morphisms. 
  The coherence isomorphisms of $\cat{C}_T$ are images of those in $\cat{C}$ under the functor $\cat{C}\to\cat{C}_T$. 
  This functor preserves daggers and hence unitaries, making all coherence isomorphisms of $\cat{C}_T$ unitary. 
  It remains to check that the dagger on $\cat{C}_T$ satisfies $(f\otimes_T g)^\dag=f^\dag\otimes_T g^\dag$.  By Theorem~\ref{thm:strong}, $T$ is isomorphic to $-\otimes T(I)$, and it is straightforward to check that this induces an isomorphism between the respective Kleisli categories that preserves daggers and monoidal structure on the nose. Thus it suffices to check that this equation holds on $\cat{C}_{-\otimes T(I)}$, which can be done with a straightforward graphical argument.
\end{proof}

\section{Frobenius--Eilenberg--Moore algebras}\label{sec:fem}

The other canonical standard categorical construction when given a monad $T$ is to consider the category $\cat{C}^T$ of its Eilenberg--Moore algebras. In monadic programming, these are understood to expand effectful computations to pure computations~\cite{jacobsheunenhasuo:arrows}. This section identifies the largest full subcategory of $\cat{C}^T$ that is still reversible.

\begin{definition}
  An \emph{Eilenberg--Moore algebra} $(A,a)$ for a monad $T$ is a morphism $T(A) \sxto{a} A$ satisfying $a \circ T(a) = a \circ \mu$ and $a \circ \eta = \id$. A \emph{morphism of Eilenberg--Moore algebras} $(A,a) \to (B,b)$ is a morphism $A \sxto{f} B$ satisfying $b \circ T(f) = f \circ a$. These form a category $\cat{C}^T$.
\end{definition}

We will again need cooperation of such algebras with daggers when present.

\begin{definition} 
   Let $T$ be a monad on a dagger category $\cat{C}$. A \emph{Frobenius--Eilenberg--Moore algebra}, or \emph{FEM-algebra} for short, is an Eilenberg--Moore algebra $(A,a)$ that makes the following diagram commute.
   \begin{equation}\label{eq:femlaw}\begin{aligned}\begin{tikzpicture}
     \matrix (m) [matrix of math nodes,row sep=1.5em,column sep=4em,minimum width=2em]
     {
      T(A) & T^2(A) & \\
      T^2(A) & T(A) \\};
     \path[->]
     (m-1-1) edge node [left] {$\mu^\dag$} (m-2-1)
             edge node [above] {$T(a)^\dag$} (m-1-2)
     (m-2-1) edge node [below] {$T(a)$} (m-2-2)
     (m-1-2) edge node [right] {$\mu$} (m-2-2);
   \end{tikzpicture}\end{aligned}\end{equation}   
   We call this the \emph{Frobenius law} for Eilenberg--Moore algebras.
\end{definition}

\begin{example}\label{ex:kleislifem}
  The Kleisli category $\cat{C}_T$ of any monad $T$ sits inside $\cat{C}^T$ as the \emph{free algebras} $(T(A),\mu_A)$.
  If $T$ is a Frobenius monad on a dagger category $\cat{C}$, any free algebra is an FEM-algebra.
\end{example}
\begin{proof}
  The Frobenius law for the free algebra is the Frobenius law of the monad.
\end{proof}

There are many EM-algebras that are not FEM-algebras; a family of examples can be derived from~\cite[Theorem~6.4]{pavlovic:abstraction}. Here is a concrete example.

\begin{example}\label{ex:emnonfem}
  Let $A=\mathbb{M}_2(\mathbb{C})$ be the Hilbert space of $2$-by-$2$-matrices, with inner product $\left\langle a,b\right\rangle=\tfrac{1}{2}\Tr (a^\dag \circ b)$. Matrix multiplication gives a map $m\colon A\otimes A\to A$ making $A$ a Frobenius monoid in $\cat{FHilb}$, so that $T=-\otimes A$ is a Frobenius monad. Let $u\in A$ be a unitary matrix, and define $U\colon A\to A$ by $U(a)=u^\dag\circ a\circ u$. Now $U^\dag(a)=u\circ a\circ u^\dag$ and $U$ is an endomorphism of the monoid $A$, making $h=m\circ (\id\otimes U)$ an EM-algebra. It is an FEM-algebra if and only if $u=u^\dag$.
\end{example}
\begin{proof}
  The Frobenius law~\eqref{eq:femlaw} means:
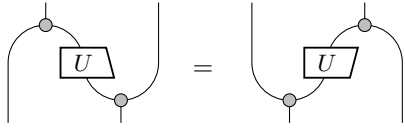
  \[
    \begin{pic}
     \draw (0,.2) to (0,1) to[out=90,in=180] (.5,1.5) to (.5,1.8);
     \draw (.5,1.5) to[out=0,in=90] (1,1) to[out=-90,in=180] (1.5,.5) to (1.5,.2);
     \draw (1.5,.5) to[out=0,in=-90] (2,1) to (2,1.8);
     \node[morphism] at (1,1) {$U$};
     \node[dot] at (.5,1.5) {};
     \node[dot] at (1.5,.5) {};
    \end{pic}
    \quad = \quad
    \begin{pic}[xscale=-1]
     \draw (0,.2) to (0,1) to[out=90,in=180] (.5,1.5) to (.5,1.8);
     \draw (.5,1.5) to[out=0,in=90] (1,1) to[out=-90,in=180] (1.5,.5) to (1.5,.2);
     \draw (1.5,.5) to[out=0,in=-90] (2,1) to (2,1.8);
     \node[morphism,hflip] at (1,1) {$U$};
     \node[dot] at (.5,1.5) {};
     \node[dot] at (1.5,.5) {};
    \end{pic}
  \]
  This comes down to $U=(U^*)^\dag$, that is, $u=u^\dag$. 
\end{proof}

The following two results highlight the importance of FEM-algebras to daggers.
First, extending from pure computations to FEM--computations is still reversible.

\begin{proposition}
  Let $T$ be a Frobenius monad on a dagger category $\cat{C}$. The dagger on $\cat{C}$ induces a dagger on the category of FEM-algebras of $T$. 
\end{proposition} 
\begin{proof}
  Let $f \colon (A,a) \to (B,b)$ be a morphism of FEM-algebras; we have to show that $f^\dag$ is a morphism $(B,b)\to(A,a)$. It suffices to show that   $b \circ T(f) = f \circ a$ implies $a \circ T(f^\dag) = f^\dag \circ b$.
  Consider the following diagram:
  
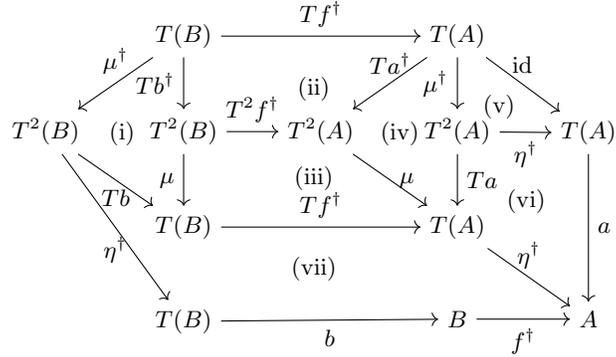
\[\begin{tikzpicture}[font=\small,scale=.6]
    \matrix (m) [matrix of math nodes,row sep=2em,column sep=2em,minimum width=1em]
    {
      & T(B) & & T(A)  \\
     T^2(B) & T^2(B) & T^2(A) & T^2(A) & T(A)  \\
      &T(B) & &T(A) &       \\      
      &T(B) & &B &A    \\};
    \path[->]
    (m-1-2) edge node [above] {$Tf^\dag$} (m-1-4)
                  edge node [above] {$\mu^\dag$} (m-2-1)
                  edge node [left] {$Tb^\dag$}(m-2-2)
    (m-1-4)  edge node [above] {$\id$} (m-2-5)
                  edge node [left] {$\mu^\dag$}(m-2-4)
                  edge node [above] {$Ta^\dag$} (m-2-3) 
    (m-2-1) edge node [below] {$Tb$} (m-3-2)
                  edge node [below] {$\eta^\dag$} (m-4-2)
    (m-2-2) edge node [left] {$\mu$} (m-3-2) 
                 edge node [above] {$T^2f^\dag$} (m-2-3)
    (m-2-3) edge node [right] {$\mu$} (m-3-4) 
    (m-2-4) edge node [right] {$Ta$} (m-3-4) 
                 edge node [below] {$\eta^\dag$} (m-2-5)
    (m-2-5) edge node [right] {$a$} (m-4-5)
    (m-3-2) edge node [above] {$Tf^\dag$} (m-3-4)
    (m-3-4) edge node [above] {$\eta^\dag$} (m-4-5)
    (m-4-2) edge node [below] {$b$} (m-4-4) 
    (m-4-4) edge node [below] {$f^\dag$} (m-4-5);
    \node at (-4.25,1) {(i)};
    \node at (0,2) {(ii)};
    \node at (0,0) {(iii)};
    \node at (1.9,1) {(iv)};
    \node at (4.1,1.6) {(v)};
    \node at (4.7,-0.5) {(vi)};
    \node at (0,-2) {(vii)};
  \end{tikzpicture}\]
  Region (i) is the Frobenius law of $(B,b)$; commutativity of (ii) follows from the assumption that $f$ is a morphism $(A,a)\to (B,b)$ by applying $T$ and $\dag$; (iii) is naturality of $\mu$; (iv) is the Frobenius law of $(A,a)$; (v) commutes since $T$ is a comonad; (vi) and (vii) commute by naturality of $\eta^\dag$. 
\end{proof}

Second, FEM--computations are the largest class that stays reversible.

\begin{theorem}
  FEM-algebras form the largest full subcategory of $\cat{C}^T$ containing $\cat{C}_T$ that carries a dagger commuting with the forgetful functor $\cat{C}^T \to \cat{C}$.
\end{theorem}
\begin{proof}
  Suppose that an EM-algebra $(A,a)$ is such that for any free algebra $(T(B),\mu_B)$ and any morphism $f :T(B)\to A$, $f$ is a morphism of EM-algebras $(T(B),\mu_B)\to (A,a)$ iff $f^\dag$ is a morphism $(A,a)\to (T(B),\mu_B)$ of EM-algebras. Now $(A,a)$ being an EM-algebra implies that $a$ is a morphism $(T(A),\mu_A)\to (A,a)$. Thus by assumption $a^\dag$ is a morphism $(A,a)\to (T(A),\mu_A)$, which implies that $(A,a)$ is an FEM-algebra.
\end{proof}

\section{Quantum measurement}\label{sec:measurement}

This final section exemplifies the relevance of FEM-algebras to quantum computation, by indicating how quantum measurement fits neatly in effectful functional programming as \emph{handlers} of Frobenius monads~\cite{plotkinpretnar:handlers,kammarlindleyoury:handlers}.

\begin{example}
  Let $B$ be a finite-dimensional Hilbert space. A choice of orthonormal basis makes $B$ a commutative Frobenius monoid in $\cat{FHilb}$ via Example~\ref{ex:cstaralgebras}. Hence $T=- \otimes B$ is a (commutative strong) Frobenius monad on $\cat{FHilb}$ by Theorem~\ref{thm:strong}.
\end{example}

Traditionally, effectful computations are modelled as morphisms in the Kleisli category~\cite{moggi:monads,wadler:monads}.
In the above example, those are just morphisms $A \to A \otimes B$ in $\cat{FHilb}$. 
Quantum measurements are indeed morphisms of this type, but they satisfy more requirements, such as von Neumann's \emph{projection postulate}: repeating a measurement is equivalent to copying the outcome of the first measurement. These requirements make the dagger of the morphism $A \to A \otimes B$ precisely an FEM-algebra, see~\cite[Theorems~1.5 and~1.6]{coeckepavlovic:measurement}.\footnote{Technically, the monad has to be lifted to a category of so-called completely positive maps, see~\cite{coeckepavlovic:measurement}.} The following proposition summarizes.

\begin{proposition}\label{prop:measurement}
  Quantum measurements with outcomes modeled by a commutative strong Frobenius monad on $\cat{FHilb}$ correspond precisely to its FEM-algebras. 
  \qed
\end{proposition}

Consider the exception monad $T$ that adds exceptions from a set $E$ to a computation by $T(A)=A+E$.
Intercepting exceptions means executing a computation $f_e$ for each $e \in E$, and a computation $f$ if no exception is raised. Thus a \emph{handler} for $T$ specifies an EM-algebra $(A,a)$ and a map $f \colon A \to \todo{A}$ making the triangle left below commute.
\[
  \begin{tikzpicture}[xscale=2, yscale=1.5]
    \node (A) at (0,1) {$A$};
    \node (AE) at (0,0) {$A + E$};
    \node (H) at (1,0) {$(A,a)$};
    \draw[->] (A) to node[left] {$\eta_A$} (AE);
    \draw[->] (A) to node[right=2mm] {$f$} (H);
    \draw[->,dashed] (AE) to (H);
  \end{tikzpicture}
  \qquad\qquad\qquad
  \begin{tikzpicture}[xscale=2, yscale=1.5]
    \node (A) at (0,1) {$A$};
    \node (AE) at (0,0) {$A \otimes B$};
    \node (H) at (1,0) {$(A,a)$};
    \draw[->] (A) to node[left] {$\eta_A$} (AE);
    \draw[->] (A) to node[right=2mm] {$f$} (H);
    \draw[->,dashed] (AE) to (H);
  \end{tikzpicture}
\]
This extends to arbitrary algebraic effects T~\cite{plotkinpretnar:handlers}.
In particular, it makes sense for quantum measurement, as in the right diagram above.
The Frobenius monad $- \otimes B$ modeling quantum measurement with outcomes in $B$ is similar to `raising exceptions $B$',
the vertical arrows are Kleisli morphisms, and the lower right \emph{handling construct} is an FEM-algebra $A \otimes B \sxto{a} A$ that `handles exceptions $B$'; it involves the unique dashed arrow, that is induced by the free property of the Kleisli algebra $A \otimes B$, and is a morphism of FEM-algebras by Example~\ref{ex:kleislifem}. 
Intuitively, Kleisli morphisms $A \to T(B)$ are constructors that `build' an effectful computation, whereas FEM algebras $T(B)\to B$ are destructors that `handle' the effects.

Thus in general, effectful reversible computation takes place in the category of FEM-algebras of a Frobenius monad, rather than its subcategory of Kleisli algebras. See also~\cite{jacobs:blocks} for a similar reasoning in different language.

\section{Conclusion}

We have \todo{proposed} Frobenius monads \todo{as} the appropriate notion to model computational effects in the reversible setting of dagger categories. We have justified their definition from first principles, characterized them internally, shown that their Kleisli categories are again reversible, and identified the largest reversible subcategory of their Eilenberg--Moore categories.
As an example we phrased quantum measurement in the category of such Frobenius--Eilenberg--Moore algebras.

More example\todo{s} should be studied. Specifically, noncommutative Frobenius monoids on $\cat{FHilb}$ might induce monads modelling partial quantum measurement.
Also, the relationship between nondeterministic computation in $\cat{Rel}$ and groupoids should be explored. 
Finally, we leave probabilistic computation to future work.

\bibliographystyle{plain}
\bibliography{reversible}

\appendix
\section{Proofs}

This appendix verifies steps used in proofs in Section~\ref{sec:monads}.

\begin{lemma}\label{lem:strictmorphismsareiso}
  A monoid homomorphism between Frobenius monoids in a monoidal dagger category, that is also a comonoid homomorphism, is an isomorphism.
\end{lemma}
\begin{proof}
  Construct an inverse to $A \sxto{f} B$ as follows:
  \[\begin{pic}[xscale=.75,yscale=.5]
      \node (f) [morphism] at (0,0) {$f$};
      \draw (f.north)
      to [out=up, in=right] +(-0.5,0.5) node (d) [dot] {}
      to [out=left, in=up] +(-0.5,-0.5)
      to [out=down, in=up] +(0,-1.8) node [below] {$B$};
      \draw (d.north) to +(0,0.4) node [dot] {};
      \draw (f.south)
      to [out=down, in=left] +(0.5,-0.5) node (d) [dot] {}
      to [out=right, in=down] +(0.5,0.5)
      to [out=up, in=down] +(0,1.8) node [above] {$A$};
      \draw (d.south) to +(0,-0.4) node [dot] {};
    \end{pic}
  \]
  The composite with $f$ gives the identity in one direction:
  
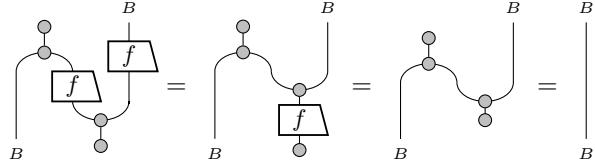
\[
    \begin{pic}[xscale=.75,yscale=.5]
      \node (f) [morphism] at (0,0) {$f$};
      \draw (f.north)
      to [out=up, in=right] +(-0.5,0.5) node (d) [dot] {}
      to [out=left, in=up] +(-0.5,-0.5)
      to [out=down, in=up] +(0,-1.8) node [below] {$B$};
      \draw (d.north) to +(0,0.5) node [dot] {};
      \draw (f.south)
      to [out=down, in=left] +(0.5,-0.5) node (d) [dot] {}
      to [out=right, in=down] +(0.5,0.5)
      to [out=down, in=down] +(0,0.8)
      node (f2) [morphism, anchor=south, width=0.2cm]
      {$f$};
      \draw (d.south) to +(0,-0.5) node [dot] {};
      \draw (f2.north) to +(0,0.5) node [above] {$B$};
    \end{pic}
    =
    \begin{pic}[xscale=.75,yscale=.5]
      \draw (0,0)
      to [out=up, in=right] +(-0.5,0.5) node (d) [dot] {}
      to [out=left, in=up] +(-0.5,-0.5)
      to [out=down, in=up] +(0,-1.8) node [below] {$B$};
      \draw (d.north) to +(0,0.5) node [dot] {};
      \node (f) [morphism] at (0.5,-1.3) {$f$};
      \draw (0,0)
      to [out=down, in=left] +(0.5,-0.5) node (e) [dot] {}
      to [out=right, in=down] +(0.5,0.5)
      to [out=up, in=down] +(0,1.3) node [above] {$B$};
      \draw (f.north) to (e.south);
      \draw (f.south) to +(0,-0.4) node [dot] {};
    \end{pic}
    =
    \begin{pic}[xscale=.75,yscale=.5]
      \draw (0,0)
      to [out=up, in=right] +(-0.5,0.5) node (d) [dot] {}
      to [out=left, in=up] +(-0.5,-0.5)
      to [out=down, in=up] +(0,-1.5) node [below] {$B$};
      \draw (d.north) to +(0,0.5) node [dot] {};
      \draw (0.5,-0.5) to +(0,-0.5) node [dot] {};
      \draw (0,0)
      to [out=down, in=left] +(0.5,-0.5) node (d) [dot] {}
      to [out=right, in=down] +(0.5,0.5)
      to [out=up, in=down] +(0,1.6) node [above] {$B$};
    \end{pic}
    =
    \begin{pic}[yscale=.5]
      \draw (0,0) node [below] {$B$} to +(0,3.1) node [above] {$B$};
    \end{pic}
  \]
  The third equality uses the Frobenius law~\eqref{eq:frobeniuslaw} and unitality.
  The other composite is the identity by a similar argument.
\end{proof}

\begin{lemma}\label{lem:froblaw} 
  The functor $T\mapsto T(I)$ preserves the Frobenius law.
\end{lemma}
\begin{proof}
  Consider the diagram in \figurename~\ref{fig:froblaw}. 
  Region (i) commutes because $T$ is a Frobenius monad, (ii) because $\mu^\dag$ is natural, (iii) because $\rho^{-1}$ is natural, (iv) because $\st^\dag$ is natural, (v) is a consequence of $T$ being a strong monad, (vi) commutes as $\rho$ is natural, (vii) and (viii) because $\st$ is natural, (ix) commutes trivially and (x) because $\st$ is natural. Regions (ii)'-(x)' commute for dual reasons. 
  Hence the outer diagram commutes, and $T\mapsto T(I)$ preserves the Frobenius law. 
\end{proof}

\begin{lemma}\label{lem:strictmorphism}
  If $T$ is a strong Frobenius monad, the counit of the adjunction of Proposition~\ref{prop:monoids} is a morphism of comonads.
\end{lemma}
\begin{proof}
  First we show that the counit of the adjunction preserves counits of the comonads. It suffices to see that
  \[\begin{tikzpicture}[font=\small]
  \matrix (m) [matrix of math nodes,row sep=2em,column sep=4em,minimum width=2em]
  {
     A\otimes T(I) & \\
     T(A\otimes I) & A\otimes I \\ 
     T(A) &  A\\};
  \path[-stealth]
    (m-1-1) edge node [left] {$\st_{A,I}$} (m-2-1)
            edge node [right] {$\ \ \id\otimes \eta^\dag_I$} (m-2-2)
    (m-2-1) edge node [below] {$\eta^\dag_{A\otimes I}$} (m-2-2)
    (m-2-1) edge node [left] {$T(\rho_A)$} (m-3-1)
    (m-2-2) edge node [right] {$\rho_A$} (m-3-2)
     (m-3-1) edge node [below] {$\eta^\dag_A$} (m-3-2);
  \end{tikzpicture}\]
  commutes. But the rectangle commutes because $\eta^\dag$ is natural, and the triangle commutes because $T$ is a strong monad and $\st$ is an isomorphism. 

  To see that that the counit of the adjunction preserves the comultiplication, consider the following diagram:
  \[\begin{tikzpicture}[font=\scriptsize,scale=.6]
    \matrix (m) [matrix of math nodes,row sep=2em,column sep=1em,minimum width=1em]
    {
     A\otimes T(I) & A\otimes T^2(I) &A\otimes T(T(I)\otimes I) &A\otimes (T(I)\otimes T(I)) \\
      & A\otimes T^2(I) &A\otimes T(T(I)\otimes I) &(A\otimes T(I))\otimes T(I) \\
      &  &T(A\otimes (T(I)\otimes I)) &T((A\otimes T(I))\otimes I) \\
      &  & &T(A\otimes T(I)) \\
     T(A\otimes I) &  & &T^2(A\otimes I) \\
     T(A) &  & &T^2(A) \\};
    \path[->]
    (m-1-1) edge node [left] {$\st$} (m-5-1)
            edge node [above] {$\id\otimes\mu^\dag$} (m-1-2)
    (m-1-2) edge node [above] {$\id\otimes T(\rho^{-1})$} (m-1-3)
    (m-1-3) edge node [above] {$\id\otimes \st^\dag$} (m-1-4)
    (m-2-3) edge node [above] {$\id\otimes T(\rho)$} (m-2-2)
    (m-1-2) edge node [right] {$\id$} (m-2-2)
    (m-1-4) edge node [right] {$\alpha$} (m-2-4)
    (m-2-4) edge node [right] {$\st$} (m-3-4)
    (m-3-4) edge node [right] {$T(\rho)$} (m-4-4)
    (m-4-4) edge node [right] {$T(\st)$} (m-5-4)
    (m-5-4) edge node [right] {$T^2(\rho)$} (m-6-4)
    (m-5-1) edge node [left] {$T(\rho)$} (m-6-1) 
    (m-6-1) edge node [below] {$\mu^\dag$} (m-6-4)
    (m-2-3) edge node [right] {$\st$} (m-3-3) 
    (m-1-4) edge node [below] {$\quad\id\otimes \st$} (m-2-3) 
    (m-3-3) edge node [above] {$T(\alpha)$} (m-3-4)
    (m-3-3) edge node [below] {$T(\id\otimes\rho)\quad\quad$} (m-4-4)
    (m-5-1) edge node [below] {$\mu^\dag$} (m-5-4)
    (m-2-2) edge [bend right=30] (m-4-4);
    \node at (-4,1) {(i)};
    \node at (0,-0.55) {$\st$};
    \node at (0,4) {(ii)};
    \node at (0,2) {(iii)};
    \node at (4,2) {(iv)};
    \node at (4,0.5) {(v)}; 
    \node at (0,-4) {(vi)};
  \end{tikzpicture}\]
  Commutativity of region (i) is a consequence of $T$ being a strong monad, and $\st$ being an iso, (ii) commutes by definition, (iii) commutes as $\st$ is natural, (iv) because $T$ is a strong functor, (v) by coherence and finally (vi) by naturality of $\mu^\dag$. Hence the outer diagram commutes, and the counit of the adjunction preserves the comultiplication.
\end{proof}

\begin{figure*}
  \centering
  \begin{sideways}
  \begin{tikzpicture}[font=\small,xscale=.8,yscale=.9]
  \matrix (m) [matrix of math nodes,row sep=3em,column sep=3em,minimum width=2em]
  {
     T(I)\otimes T(T(I)\otimes I) & T(I)\otimes (T(I)\otimes T(I)) & & (T(I)\otimes T(I))\otimes T(I) & T(T(I)\otimes I)\otimes T(I)) \\
     &T(T(I)\otimes (T(I)\otimes I)) & T((T(I)\otimes T(I))\otimes I) & T(T(I)\otimes I)\otimes T(I) & \\
     T(I)\otimes T^2(I) & T(T(I)\otimes T(I)) & T(T(T(I)\otimes I)\otimes I) & T(T^2(I)\otimes I) & T^2(I)\otimes T(I) \\
     &T^2(T(I)\otimes I) &T^3(I) & T(T(I)\otimes I)& \\ 
     T(I)\otimes T(I) & T^2(I) & &T^2(I) & T(I)\otimes T(I) \\
     & T(T(I)\otimes I) & T^3(I) & T^2(T(I)\otimes I) & \\
     T^2(I)\otimes T(I) & T(T^2(I)\otimes I) & T(T(T(I)\otimes I)\otimes I) & T(T(I)\otimes T(I)) & T(I)\otimes T^2(I) \\
     & T(T(I)\otimes I)\otimes T(I) & T((T(I)\otimes T(I))\otimes I) & T(T(I)\otimes (T(I)\otimes I)) & \\
     T(T(I)\otimes I)\otimes T(I) & (T(I)\otimes T(I))\otimes T(I) & & T(I)\otimes (T(I)\otimes T(I)) & T(I)\otimes T(T(I)\otimes I) \\};
   \path[->]
    (m-1-1) edge node [above] {$\id\otimes \st^\dag$} (m-1-2)
    (m-1-2) edge node [above] {$\alpha$} (m-1-4)
    (m-1-4) edge node [above] {$\st\otimes\id$} (m-1-5)
    (m-1-5) edge node [right] {$T(\rho)\otimes\id$} (m-3-5)
    (m-1-1) edge node [above] {$\st$} (m-2-2)
    (m-2-2) edge node [above] {$T(\alpha)$} (m-2-3)
    (m-2-3) edge node [above] {$\st^\dag$} (m-1-4)
    (m-2-4) edge node [right] {$\st^\dag\otimes\id$} (m-1-4) 
    (m-2-4) edge node [above] {$\quad\quad T(\rho)\otimes \id$} (m-3-5)
    (m-3-1) edge node [above] {$\st$} (m-3-2)
    (m-3-1) edge node [left] {$\id\otimes T(\rho^{-1})$} (m-1-1)
    (m-3-2) edge node [left] {$T(\id\otimes\rho^{-1})$} (m-2-2)
    (m-3-2) edge node [below] {$\quad T(\rho^{-1})$} (m-2-3)
    (m-3-3) edge node [above] {$T(T(\rho)\otimes\id)$} (m-3-4)
    (m-3-3) edge node [right] {$T(\st^\dag\otimes\id)$} (m-2-3)
    (m-3-3) edge node [right] {$\st^\dag$} (m-2-4)
    (m-3-4) edge node [above] {$\st^\dag$} (m-3-5)
    (m-3-4) edge node [right] {$T(\mu\otimes\id)$} (m-4-4)
    (m-3-5) edge node [right] {$\mu\otimes\id$} (m-5-5)
    (m-4-2) edge node [left] {$T(\st^\dag)$} (m-3-2)
    (m-4-2) edge node [right] {$\quad T(\rho^{-1})$} (m-3-3)
    (m-4-2) edge node [above] {$T^2(\rho)$} (m-4-3)
    (m-4-3) edge node [above] {$T(\mu_I)$} (m-5-4)
    (m-4-4) edge node [above] {$\st^\dag$} (m-5-5)
    (m-5-1) edge node [above] {$\st$} (m-6-2)
    (m-5-1) edge node [left] {$\id\otimes\mu^\dag$} (m-3-1)
    (m-5-1) edge node [left] {$\mu^\dag\otimes\id$} (m-7-1)
    (m-5-2) edge node [above] {$\mu^\dag_{T(I)}$} (m-4-3)
    (m-5-2) edge node [above] {$T(\mu^\dag_I)$} (m-6-3)
    (m-5-4) edge node [left] {$T(\rho^{-1})$} (m-4-4)
    (m-6-2) edge node [right] {$T(\rho)$} (m-5-2)
    (m-6-2) edge node [left] {$T(\mu^\dag\otimes\id)$} (m-7-2)
    (m-6-3) edge node [above] {$\mu_{T(I)}$} (m-5-4)
    (m-6-3) edge node [above] {$\quad T^2(\rho^{-1})$} (m-6-4)
    (m-7-1) edge node [above] {$\st$} (m-7-2)
    (m-7-1) edge node [left] {$T(\rho^{-1})\otimes\id$} (m-9-1)
    (m-7-1) edge node [above] {$\ \ \quad T(\rho^{-1})\otimes\id$} (m-8-2)
    (m-7-2) edge node [above] {$T(T(\rho^{-1})\otimes\id)$} (m-7-3)
    (m-7-3) edge node [above] {$T(\rho)$} (m-6-4)
    (m-7-4) edge node [right] {$T(\st)$} (m-6-4)
    (m-7-4) edge node [above] {$\st^\dag$} (m-7-5)
    (m-7-5) edge node [right] {$\id\otimes\mu$} (m-5-5)
    (m-8-2) edge node [above] {$\st$} (m-7-3)
    (m-8-3) edge node [right] {$T(\st\otimes\id)$} (m-7-3)
    (m-8-3) edge node [right] {$\quad T(\rho)$} (m-7-4)
    (m-8-3) edge node [above] {$T(\alpha^{-1})$} (m-8-4)
    (m-8-4) edge node [right] {$T(\id\otimes\rho)$} (m-7-4)
    (m-8-4) edge node [above] {$\ \ \st^\dag$} (m-9-5)
    (m-9-1) edge node [below] {$\st^\dag\otimes\id $} (m-9-2)
    (m-9-2) edge node [left] {$\st\otimes\id $} (m-8-2)
    (m-9-2) edge node [above] {$\st$} (m-8-3)
    (m-9-2) edge node [below] {$\alpha^{-1}$} (m-9-4)
    (m-9-4) edge node [below] {$\id\otimes \st$} (m-9-5)
    (m-9-5) edge node [right] {$\id\otimes T(\rho)$} (m-7-5)
    (m-6-2) edge [bend left=30] (m-4-2)
    (m-6-4) edge [bend right=30] (m-4-4);
    \node at (-6.5,0) {$\mu^\dag$};
    \node at (6.25,0) {$\mu$}; 
    \node at (0,0) {(i)}; 
    \node at (-5,1) {(ii)}; 
    \node at (5,-1) {(ii)'}; 
    \node at (3,2) {(iii)}; 
    \node at (-3,-2) {(iii)'};
    \node at (8,2.5) {(iv)}; 
    \node at (-8,-2.5) {(iv)'}; 
    \node at (-8,1.5) {(v)'}; 
    \node at (8,-1.5) {(v)}; 
    \node at (2.75,-3.5) {(vi)}; 
    \node at (-2.75,3.5) {(vi)'};
    \node at (-2.75,-5.25) {(vii)}; 
    \node at (2.75,5.25) {(vii)'}; 
    \node at (-5,-4.5) {(viii)}; 
    \node at (5,4.5) {(viii)'};
    \node at (-8,-6) {(ix)}; 
    \node at (8,6) {(ix)'};
    \node at (-8,5) {(x)};
    \node at (8,-5) {(x)'}; 
    \node at (-4.25,4.75) {(xi)}; 
    \node at (4.5,-4.75) {(xi)'};
    \node at (-2.65,6.5) {(xii)}; 
    \node at (2.65,-6.5) {(xii)'};
  \end{tikzpicture}
  \end{sideways}
  \caption{Diagram proving that $T\mapsto T(I)$ preserves the Frobenius law.}
  \label{fig:froblaw}
\end{figure*}

\begin{corollary}\label{cor:special}
  The equivalence of Theorem~\ref{thm:strong} restricts to an equivalence between special Frobenius monoids and special strong Frobenius monads.
\end{corollary}
\begin{proof}
  The commutative case follows from~\cite{wolff:monads}. 
  If the Frobenius monoid is special, so is the monad, by trivial graphical manipulation of Example~\ref{example:tensor}. Conversely, if the Frobenius monad $T$ is special, the following diagram commutes:
  \[\begin{tikzpicture}
  \matrix (m) [matrix of math nodes,row sep=2em,column sep=2.5em,minimum width=1.5em,font=\small]
    { T(I) & T^2(I) & T\big( T(I) \otimes I \big) & T(I) \otimes T(I) \\
      T(I) & T^2(I) & T\big( T(I) \otimes I \big) \\};
    \path[->]
    (m-1-1) edge node [left] {$\id$} (m-2-1)
            edge node [above] {$\mu^\dag$} (m-1-2)
    (m-1-2) edge node [left] {$\id$} (m-2-2)
            edge node [above] {$T(\rho^{-1})$} (m-1-3)
    (m-1-3) edge node [left] {$\id$} (m-2-3)
            edge node [above] {$\st^{-1}$} (m-1-4)
    (m-1-4) edge node [below] {$\st$} (m-2-3)
    (m-2-3) edge node [below] {$T(\rho)$} (m-2-2)
    (m-2-2) edge node [below] {$\mu$} (m-2-1);
  \end{tikzpicture}\]
  and so $T(I)$ is special.
\end{proof}

\end{document}